\documentclass[11pt]{article}

\usepackage[margin=1in,headheight=14pt,footskip=30pt]{geometry}
\usepackage{amsmath,amssymb,amsthm}
\usepackage{array}
\usepackage{microtype}
\usepackage[hidelinks]{hyperref}
\hypersetup{
  pdftitle={Step Recursion: Exact Depth Does Not Determine Algebraic Expressiveness},
  pdfauthor={Kirill Osipov},
  pdfsubject={Exact-depth separation and algebraic multiplicity for bounded step recursion},
  pdfkeywords={step recursion, bounded recursion, generalized inverse, depth equivalence, function algebra, Grzegorczyk hierarchy}
}

\newcommand{\Nat}{\mathbb N}
\newcommand{\Alg}[2]{\mathcal A_{#1}(#2)}
\newcommand{\Halv}[1]{\mathcal B_{0,#1}}
\newcommand{\SR}{\operatorname{SR}}
\newcommand{\blen}{\lambda}

\newcommand{\DepthEq}{\equiv_D}

\newtheorem{theorem}{Theorem}
\newtheorem{corollary}[theorem]{Corollary}
\newtheorem{lemma}[theorem]{Lemma}
\newtheorem{proposition}[theorem]{Proposition}
\newtheorem{definition}[theorem]{Definition}
\newtheorem{remark}[theorem]{Remark}

\title{Step Recursion:\\Exact Depth Does Not Determine Algebraic Expressiveness}
\author{Kirill Osipov\\
{\normalsize Independent researcher, The Hague, The Netherlands}\\
{\normalsize \href{mailto:d503@acmer.me}{\texttt{d503@acmer.me}}}}
\date{August 2026}

\begin{document}
\maketitle

\begin{abstract}
Step recursion is a form of bounded recursion in which each recursive call
moves from an input $y$ to a prescribed predecessor $\rho_g(y)$.  Its depth
$D_g(y)$ is the exact number of such moves needed to reach zero.  A natural
question is whether knowing this depth for every input determines the
expressive power of the resulting function algebra.  We prove that it does
not.

We first construct two simple generators with exactly the same depth map but
different step-recursion algebras.  One gives ordinary binary halving,
$b(x)=2x+1$; the other is
\[
 p(0)=1,\qquad p(x)=x+2^{\lambda(x)}\quad(x>0),
\]
where $\lambda$ is binary length.  Although $D_p=D_b$ pointwise, the
predecessor $\rho_p$ cannot be defined from any fixed-stride binary-halving
descent at basis zero.  Thus two recursion schemes may take exactly the same
number of steps on every input and still have different expressive power.

The phenomenon is much larger than this example.  Whenever infinitely many
depth levels allow more than one predecessor arrangement, a single exact depth
profile supports $2^{\aleph_0}$ distinct step-recursion algebras over every
countable basis containing zero and the projections.  In the computable
setting the corresponding effective family has exactly $\aleph_0$ distinct
algebras.  Hence exact recursion depth is an informative resource measure, but
it is not a complete invariant: the geometry of predecessor choices inside
each depth level carries additional algebraic information.
\end{abstract}



\section{Introduction}

Function algebras characterize computation through initial functions and
closure schemes; classical examples include the Grzegorczyk hierarchy,
bounded recursion, recursion on notation, and safe or ramified recursion
\cite{Grzegorczyk1953,Cobham1965,BellantoniCook1992,Leivant1994,Rose1984}.
Bounded step recursion isolates a different parameter: the predecessor
schedule.  For a strictly increasing generator $g:\Nat\to\Nat$ with
$g(x)\ge x+1$, put
\[
 \rho_g(0)=0,\qquad \rho_g(y)=\min\{z:g(z)\ge y\}\quad(y>0),
\]
and let $D_g(y)$ be the number of iterations of $\rho_g$ needed to reach zero.

Step Recursion was introduced in
\cite{Osipov2026Step}, and generalized-inverse depth profiles and their orbit
semantics were developed in \cite{Osipov2026Resource}.  Those results make
depth a meaningful resource coordinate but do not decide whether it becomes a
complete invariant when known pointwise.  The resource-profile work studies the
depth/orbit data themselves; the present paper holds that entire data fixed and
asks how much algebraic multiplicity remains inside a single exact-depth fibre.
Thus we freeze the entire depth map, equivalently every orbit endpoint, and vary
only the parent assignment inside the resulting levels.  We ask
\[
 \boxed{D_g=D_h\quad\Longrightarrow\quad
        \Alg{\rho_g}{B}=\Alg{\rho_h}{B}\ ?}
\]
with literal pointwise equality.  A negative answer therefore separates
algebras strictly below the resolution of exact depth, rather than by changing
asymptotic recursion length or a fixed stride parameter.  This is qualitatively
stronger than the earlier fixed-stride separations: equality $D_g=D_h$ fixes
not merely an asymptotic depth scale but every stage count and, equivalently,
every orbit endpoint; the remaining freedom lies entirely inside the depth
levels.

\paragraph{Main results.}
First, at basis zero we give an explicit low-complexity separation.  Let
$b(x)=2x+1$, so $\rho_b(y)=\lfloor y/2\rfloor$, and define
\[
 p(0)=1,\qquad p(x)=x+2^{\lambda(x)}\quad(x>0),
\]
where $\lambda$ is binary length.  Then
$p^{[k]}(0)=b^{[k]}(0)=2^k-1$, hence $D_p=D_b$, but a source--offset normal
form for fixed-stride halving yields
\[
 \boxed{\rho_p\notin\Halv{l}\ (l\ge1)},\qquad
 \boxed{\Alg{\rho_p}{B_0}\ne\Alg{\rho_b}{B_0}}.
\]
Here $\Halv{l}$ is the basis-zero closure along $\rho(y)=\lfloor y/2^l\rfloor$;
this is not the three-parameter class $\mathcal H^0_{1,l}$ of
\cite{Osipov2026Step}.
The witness $p$ is a linear-time binary transduction, so the separation is not
an artefact of an intractable generator.

Second, the phenomenon is large.  A level-assignment lemma shows that if
infinitely many levels are flexible, then for every countable basis $B$
containing zero and the projections,
\[
 \boxed{\mathfrak M_B(D)=2^{\aleph_0}}.
\]
For computable $D$ and an effective basis, the computable subfamily has exactly
\[
 \boxed{\mathfrak M_B^{\rm eff}(D)=\aleph_0},
\]
and may be chosen so that no later algebra is contained in an earlier one.
Thus exact depth is a genuine resource coordinate but not a semantic
classification: a single exact-depth fibre can already have maximal algebraic
multiplicity.  The missing information is within-level predecessor geometry.

\section{Framework}

Throughout, $\Nat=\{0,1,2,\ldots\}$.  All functions are total and of positive
finite arity.  Let $B_m$ denote the level-$m$ initial basis used for the
Grzegorczyk classes: zero, successor, all projections, and the first $m$
Grzegorczyk generators in the standard indexing of \cite{Rose1984}.
Only two facts about these bases will be needed here: $B_0$ contains zero,
successor, and projections, and $B_m\subseteq B_{m+1}$.

\begin{definition}[bounded step recursion]
A \emph{descent} is a nondecreasing map $\rho:\Nat\to\Nat$ with
\[
 \rho(0)=0,\qquad \rho(y)<y\quad(y>0).
\]
Given earlier functions $a(\bar x)$, $s(\bar x,z,w)$, and a pointwise bound
$q(\bar x,y)$, bounded step recursion along $\rho$ forms $f$ by
\[
 f(\bar x,0)=a(\bar x),\qquad
 f(\bar x,y)=s\bigl(\bar x,\rho(y),f(\bar x,\rho(y))\bigr)\quad(y>0),
\]
provided $f(\bar x,y)\le q(\bar x,y)$ for all inputs.  For any initial basis
$B$, write
\[
 \Alg{\rho}{B}
 =\operatorname{Cl}_{\circ,\SR_\rho}(B)
\]
for finite-stage closure under composition and bounded step recursion.  The
principal examples below use $B=B_m$.
\end{definition}

\begin{definition}[generator and depth]
A \emph{generator} is a strictly increasing $g:\Nat\to\Nat$ satisfying
$g(x)\ge x+1$.  Its generalized inverse is
\[
 \rho_g(0)=0,\qquad
 \rho_g(y)=\min\{z:g(z)\ge y\}\quad(y>0),
\]
and its depth map is
\[
 D_g(y):=D_{\rho_g}(y)
 =\min\{t:\rho_g^{[t]}(y)=0\}.
\]
The generalized inverse is indeed a descent.  It is nondecreasing because the
threshold set $\{z:g(z)\ge y\}$ moves to the right as $y$ increases, and for
$y>0$ the inequality $g(y-1)\ge y$ gives
\[
 \rho_g(y)\le y-1<y.
\]
The threshold law is
\begin{equation}
 \label{eq:threshold}
 D_g(y)\le k
 \quad\Longleftrightarrow\quad
 y\le g^{[k]}(0).
\end{equation}
\end{definition}

For completeness, the threshold law follows by induction on $k$.  The case
$k=0$ is immediate.  For $y>0$, the generalized-inverse identity
\[
 \rho_g(y)\le a\quad\Longleftrightarrow\quad y\le g(a)
\]
and the induction hypothesis give
\[
 D_g(y)\le k+1
 \quad\Longleftrightarrow\quad
 D_g(\rho_g(y))\le k
 \quad\Longleftrightarrow\quad
 \rho_g(y)\le g^{[k]}(0)
 \quad\Longleftrightarrow\quad
 y\le g^{[k+1]}(0).
\]
Thus $g^{[k]}(0)$ is the final input of depth at most $k$; in the canonical
breadth-first tree it is the final vertex of that depth prefix.

\begin{proposition}[exact depth is exact orbit data]
\label{prop:depth-orbit}
For generators $g,h$ the following are equivalent:
\[
 D_g=D_h,
 \qquad
 g^{[k]}(0)=h^{[k]}(0)\quad\text{for every }k\ge0.
\]
\end{proposition}

\begin{proof}
By the threshold law, the set of inputs of depth at most $k$ is exactly
$\{0,1,\ldots,g^{[k]}(0)\}$.  Equality of depth maps therefore gives equality
of the endpoint $g^{[k]}(0)$ for every $k$.  Conversely, equality of all orbit
endpoints makes the threshold sets identical for every $k$, hence the depth
maps are identical.
\end{proof}

Proposition~\ref{prop:depth-orbit} is the only orbit-level fact needed in the
sequel: fixing the exact depth map is equivalent to fixing all orbit endpoints.

\begin{definition}[exact-depth equivalence]
For generators $g,h$, write
\[
 g\DepthEq h
 \quad\Longleftrightarrow\quad
 D_g=D_h.
\]
\end{definition}

The relation forgets all local parent information inside each depth level while
retaining the exact number of predecessor stages at every input.

\paragraph{Elementary facts.}
The following three facts will be used repeatedly.

\begin{lemma}[the descent is internal]
\label{lem:descent-internal}
Let $B$ contain zero and the projections.  For every descent $\rho$,
\[
 \rho\in\Alg{\rho}{B}.
\]
\end{lemma}

\begin{proof}
Define a binary function $F(x,y)$ by bounded step recursion along $\rho$ with
base $F(x,0)=0$ and transition
\[
 F(x,y)=P_2^3\bigl(x,\rho(y),F(x,\rho(y))\bigr)=\rho(y)
 \qquad(y>0).
\]
Use the projection $P_2^2(x,y)=y$ as a bound.  Then
$\rho(y)=F(0,y)$ is obtained by composition with zero and projections.
\end{proof}

\begin{lemma}[countability]
\label{lem:countable}
For every fixed descent $\rho$ and every countable initial basis $B$, the
algebra $\Alg{\rho}{B}$ is countable.
\end{lemma}

\begin{proof}
At each finite closure stage, only finitely many previously constructed
functions are used in each composition or bounded step-recursion term, so a
countable family produces a countable next stage.  The union of the countably
many finite stages is countable.
\end{proof}

\begin{lemma}[a generator is determined by its descent]
\label{lem:generator-injective}
For every generator $g$ and every $x\in\Nat$,
\[
 g(x)=\max\{y:\rho_g(y)\le x\}.
\]
Consequently the map $g\mapsto\rho_g$ is injective on generators.
\end{lemma}

\begin{proof}
The generalized-inverse identity gives
\[
 \rho_g(y)\le x\quad\Longleftrightarrow\quad y\le g(x).
\]
Hence the set on the right is exactly $\{0,1,\ldots,g(x)\}$, whose maximum is
$g(x)$.  Thus $\rho_g$ determines every value of $g$.
\end{proof}

\section{An explicit exact-depth separation}
\label{sec:explicit-separation}

Let
\[
 b(x)=2x+1,
 \qquad
 \rho_b(y)=\left\lfloor\frac y2\right\rfloor.
\]
For $x>0$, let
\[
 \blen(x)=1+\lfloor\log_2x\rfloor
\]
be its binary length.  Define the \emph{plateau generator}
\begin{equation}
 \label{eq:plateau-generator}
 p(0)=1,\qquad
 p(x)=x+2^{\blen(x)}\quad(x>0).
\end{equation}

On each dyadic block $2^{k-1}\le x\le2^k-1$, $p(x)=x+2^k$; in binary this
prefixes one $1$, so $p$ is linear-time computable.  Its first parent in each
nontrivial level receives a large child block while the remaining parents
receive one child each.

\begin{lemma}[same orbit, same exact depth]
\label{lem:plateau-depth}
The function $p$ is a generator and, for every $k\ge0$,
\[
 p^{[k]}(0)=2^k-1=b^{[k]}(0).
\]
Consequently
\[
 \boxed{D_p(y)=D_b(y)=\left\lceil\log_2(y+1)\right\rceil.}
\]
\end{lemma}

\begin{proof}
Inside a dyadic block $[2^{k-1},2^k-1]$, the map $p$ increases by one.  Across
the boundary,
\[
 p(2^k-1)=2^{k+1}-1
 <3\cdot2^k=p(2^k),
\]
so $p$ is strictly increasing and clearly satisfies $p(x)\ge x+1$.
Moreover, if $x=2^k-1$, then $\blen(x)=k$ for $k\ge1$, hence
\[
 p(2^k-1)=2^{k+1}-1.
\]
Together with $p(0)=1$, induction gives the orbit identity.  The depth formula
now follows from the threshold law \eqref{eq:threshold}.
\end{proof}

The generalized inverse of $p$ is explicit.

\begin{lemma}[plateau predecessor]
\label{lem:plateau-rho}
For $k\ge1$ and $2^k\le y\le2^{k+1}-1$,
\[
 \rho_p(y)=
 \begin{cases}
 2^{k-1},&2^k\le y\le3\cdot2^{k-1},\\[2mm]
 y-2^k,&3\cdot2^{k-1}<y\le2^{k+1}-1.
 \end{cases}
\]
Also $\rho_p(1)=0$.
\end{lemma}

\begin{proof}
The possible parents of the displayed level are the integers
$2^{k-1},\ldots,2^k-1$.  Immediately before the first one,
\[
 p(2^{k-1}-1)=2^k-1,
\]
whereas
\[
 p(2^{k-1})=3\cdot2^{k-1}.
\]
Thus all integers from $2^k$ through $3\cdot2^{k-1}$ have the same parent
$2^{k-1}$.  After that point, $p(z)=z+2^k$ throughout the rest of the parent
level, so generalized inversion gives $z=y-2^k$.
\end{proof}

We now prove the symbolic normal form needed for the separation.  The proof is
split into a qualitative conservation lemma and a quantitative evaluation
lemma.  This makes explicit the two facts used later: no derivation can merge
two independent large root sources into one output, and the source-free
numerical correction created by a fixed derivation is only polynomially large.

Fix $l\ge1$, put $B=2^l$, and write
\[
 \rho_l(y)=\left\lfloor\frac yB\right\rfloor.
\]
Give a derivation its \emph{construction rank}: initial-function leaves have rank zero, and a composition or step-recursion node has rank one plus the maximum rank of its proper subderivations.

For a fixed root tuple $\bar x=(x_1,\ldots,x_r)$, a \emph{descriptor} is either
\[
 \mathbf C(c)
 \qquad\text{or}\qquad
 \mathbf S(i,a,d),
\]
with represented values
\[
 c,
 \qquad
 \left\lfloor\frac{x_i}{B^a}\right\rfloor+d.
\]
All descriptor fields are nonnegative integers.

\begin{lemma}[one-source conservation]
\label{lem:one-source-conservation}
Fix a derivation $\delta$ over $B_0$ using $\rho_l$ and a concrete root tuple
$\bar x$.  Every value occurrence in the complete evaluation of $\delta$ admits
a descriptor relative to the same root tuple.  In particular, each occurrence
contains at most one root source $x_i$.
\end{lemma}

\begin{proof}
Proceed by induction on construction rank and, within a rank, on syntax.  Zero
produces $\mathbf C(0)$, a projection passes one existing descriptor, and
successor increments its constant or offset field.  The descent preserves the
two descriptor types exactly:
\[
 \rho_l(\mathbf C(c))=\mathbf C(\lfloor c/B\rfloor)
\]
and
\[
 \rho_l(\mathbf S(i,a,d))
 =\mathbf S\!\left(i,a+1,
   \left\lfloor\frac{r_{i,a}+d}{B}\right\rfloor\right),
 \qquad
 r_{i,a}:=\left\lfloor\frac{x_i}{B^a}\right\rfloor\bmod B.
\]
Thus a descent cannot introduce a second root source.

At a composition node, evaluate the argument derivations first and regard their
descriptors as the concrete inputs of the fixed head derivation.  By the
induction hypothesis for that head, its output selects at most one of those
formal inputs as a source; substituting the corresponding argument descriptor
therefore still leaves at most one original root source.

At a step-recursion node, the outer parameters and each address occurrence
already have descriptors.  Induct along the finite $\rho_l$-schedule.  The
base value has a descriptor by the construction-rank hypothesis.  At a later
stage the transition is a fixed proper subderivation applied to parameter,
address, and previous-state descriptors, so the same hypothesis again returns
a descriptor with at most one root source.  This proves the claim for all
nested recursion ranks.
\end{proof}

\begin{lemma}[logarithmic self-bound]
\label{lem:log-self-bound}
For every fixed $K\ge1$ and integer $e\ge1$ there is $C\ge1$ such that, for all
$N\ge2$ and $S\ge1$,
\[
 S\le K\bigl(N+\log_2(S+2)\bigr)^e
 \quad\Longrightarrow\quad
 S\le N^C.
\]
\end{lemma}

\begin{proof}
Choose $S_0$ so large that
$K(2\log_2(S+2))^e<S$ for every $S\ge S_0$.  If
$S>2^N$ and $2^N\ge S_0$, then $N<\log_2 S$ and the assumed inequality gives
\[
 S\le K\bigl(2\log_2(S+2)\bigr)^e<S,
\]
a contradiction.  Hence for all sufficiently large $N$ one has $S\le2^N$.
Substitution back into the hypothesis gives
\[
 S\le K(2N+2)^e,
\]
which is bounded by $N^C$ after increasing $C$.  The finitely many smaller
values of $N$ are absorbed by one further increase of $C$.
\end{proof}

For the quantitative estimate, a \emph{primitive symbolic operation} is one
local action in the fully expanded symbolic evaluation: evaluation of an
initial-function occurrence, one application of $\rho_l$, or one descriptor
substitution or normalization associated with composition or a recursion
stage.  As in the semantics of bounded recursion, the declared bound is an
admissibility certificate and is not evaluated as part of the operational
recurrence; if the same earlier function is used elsewhere as ordinary data,
that occurrence is counted normally.  A proper subderivation contributes all
of its own primitive operations, rather than one additional aggregate unit.
Any equivalent fixed local accounting convention changes the count only by a
derivation-dependent constant factor.

\begin{lemma}[quantitative symbolic evaluation]
\label{lem:quantitative-symbolic}
Fix a derivation $\delta$ over $B_0$ using $\rho_l$.  On a root tuple of total
binary length $N\ge2$, let $S$ be the number of primitive symbolic operations
in its complete evaluation.  There are constants $K_\delta,e_\delta,A_\delta$
such that
\[
 S\le K_\delta\bigl(N+\log_2(S+2)\bigr)^{e_\delta},
\]
and every descriptor field occurring in the evaluation is at most
$A_\delta+S$.
\end{lemma}

\begin{proof}
The field bound follows directly from the primitive rules.  A projection only
passes fields, successor increments one field, and one descent increments $a$
by one while replacing $d$ by
$\lfloor(r_{i,a}+d)/B\rfloor$; zero or a source-free reset contributes only a
fixed constant.  Composition does not create a hidden larger field.  Indeed,
substituting a source descriptor into another source descriptor gives
\[
 \left\lfloor
   \frac{\left\lfloor x_i/B^a\right\rfloor+d}{B^b}
 \right\rfloor+e
 =
 \left\lfloor\frac{x_i}{B^{a+b}}\right\rfloor+d',
\]
where
\[
 0\le d'\le \left\lfloor\frac d{B^b}\right\rfloor+e+1\le d+e+1.
\]
Substitution into a constant descriptor remains source-free.  Thus composing
or nesting descriptors adds exponents and offsets at most linearly in the
fields already present; after charging the finitely many local normalization
steps to the symbolic evaluation, no primitive symbolic operation increases a
field by more than a derivation-dependent constant.  This is absorbed into
$A_\delta+S$.

The preceding calculation also rules out a hidden blow-up at composition
nodes.  After an argument descriptor is substituted into the unique formal
source selected by the head derivation, the result still stores only one root
source index together with exponent and offset fields bounded by
$A_\delta+S$; the full numerical value of that root source is never copied
into a descriptor field.  Consequently every non-source field needs only
$O_\delta(\log(S+2))$ bits, while access to the surviving root source costs at
most its original $N$ input bits.  This is the point at which the self-bound
below depends on $N+\log(S+2)$ rather than on an uncontrolled intermediate
value.

Put
\[
 L^*=N+\log_2(S+2).
\]
Every descriptor encountered in the evaluation has binary length
$O_\delta(L^*)$.  We now induct on construction rank.  Initial functions and
composition nodes contribute only fixed finite combinations of lower-rank
subevaluations, hence polynomially many primitive operations in $L^*$ by the
induction hypothesis.  At a step-recursion node, every recursion address has
bitlength $O_\delta(L^*)$, so repeated division by the fixed $B=2^l$ reaches
zero in $O_\delta(L^*)$ stages.  Each stage invokes only finitely many fixed
proper subderivations, all of lower construction rank and all on descriptors of
bitlength $O_\delta(L^*)$.  Multiplying their polynomial bounds by the
$O_\delta(L^*)$ schedule length still gives a polynomial in $L^*$.  Taking the
maximum exponent and constant over the finitely many nodes of $\delta$ yields
the displayed inequality.
\end{proof}

\begin{proposition}[source--offset normal form for fixed-stride halving]
\label{prop:source-offset}
Fix $l\ge1$ and let $\rho_l(y)=\lfloor y/2^l\rfloor$.  For every fixed
derivation $\delta$ over $B_0$ using $\rho_l$, there is a constant
$C_{\delta,l}$ such that on every root input tuple
$\bar x=(x_1,\ldots,x_r)$ of total binary length $N\ge2$, the output has one of
the forms
\[
 c,
 \qquad\text{or}\qquad
 \left\lfloor\frac{x_i}{2^{la}}\right\rfloor+d
\]
for some root coordinate $i$, where
\[
 a,c,d\le N^{C_{\delta,l}}.
\]
The choice of form, source coordinate, and numerical fields may depend on the
input tuple.  Consequently every function in
$\Halv{l}:=\Alg{\rho_l}{B_0}$ has a pointwise source--offset normal form
with a single root source and polynomially bounded fields.
\end{proposition}

\begin{proof}
Lemma~\ref{lem:one-source-conservation} gives the two descriptor types for the
output occurrence.  Lemma~\ref{lem:quantitative-symbolic} gives the
self-bound for the number $S$ of primitive symbolic operations, and
Lemma~\ref{lem:log-self-bound} turns it into $S\le N^{C}$ for a fixed $C$.
The field estimate in Lemma~\ref{lem:quantitative-symbolic} then makes
$a,c,d$ polynomially bounded as well.  Enlarging the exponent absorbs all
fixed constants and gives the stated $C_{\delta,l}$.
\end{proof}

\begin{remark}[structural content of the normal form]
\label{rem:normal-form-content}
The proposition is structural, not merely a running-time bound: on each input a
basis-zero fixed-stride derivation has at most one large root source, followed
by a fixed-base quotient and a polynomially bounded correction.  The plateau
predecessor violates exactly this constraint.
\end{remark}

\begin{theorem}[explicit exact-depth separation]
\label{thm:effective-separation}
The explicit computable generators $b$ and $p$ satisfy
\[
 b\DepthEq p,
\]
but the predecessor $\rho_p$ is excluded from every basis-zero fixed-stride
binary class:
\[
 \boxed{\rho_p\notin\Halv{l}\qquad(l\ge1).}
\]
In particular,
\[
 \boxed{\rho_p\notin\Alg{\rho_b}{B_0}},
 \qquad
 \boxed{\Alg{\rho_p}{B_0}\ne\Alg{\rho_b}{B_0}.}
\]
\end{theorem}

\begin{proof}
Exact depth equivalence is Lemma~\ref{lem:plateau-depth}.  Fix $l\ge1$, put
$B=2^l$, and suppose toward a contradiction that
$\rho_p\in\Halv{l}$.  Let $C$ be the constant supplied by
Proposition~\ref{prop:source-offset} for a fixed derivation of this unary
function, and put
\[
 x_k=3\cdot2^{k-1}\qquad(k\ge1).
\]
Its binary length is $N_k=k+1$, and Lemma~\ref{lem:plateau-rho} gives
\[
 \rho_p(x_k)=2^{k-1}.
\]

For sufficiently large $k$, this value cannot be represented by a source-free
descriptor, because $2^{k-1}>(k+1)^C$.  Hence
\[
 2^{k-1}
 =\left\lfloor\frac{3\cdot2^{k-1}}{2^{la}}\right\rfloor+d
\]
for some $a,d\le(k+1)^C$.
If $a=0$, the shifted source already exceeds the target.  If $la=1$, then
necessarily $l=a=1$ and
\[
 \left\lfloor\frac{x_k}{2}\right\rfloor
 =3\cdot2^{k-2}>2^{k-1},
\]
again impossible because $d\ge0$.  Finally, if $la\ge2$, then
\[
 \left\lfloor\frac{x_k}{2^{la}}\right\rfloor
 \le3\cdot2^{k-3},
\]
so
\[
 d\ge2^{k-1}-3\cdot2^{k-3}=2^{k-3},
\]
contradicting $d\le(k+1)^C$ for large $k$.  Therefore
$\rho_p\notin\Halv{l}$ for every fixed $l\ge1$.

Taking $l=1$ gives $\rho_p\notin\Alg{\rho_b}{B_0}$.  By
Lemma~\ref{lem:descent-internal}, $\rho_p\in\Alg{\rho_p}{B_0}$, so the two
algebras are different.
\end{proof}

\section{Exact-depth fibres and algebraic multiplicity}
\label{sec:flexible-fibers}

\begin{definition}[level sizes and flexibility]
Let $D$ be a depth map realised by at least one generator.  Put
\[
 A_k(D):=\max\{y:D(y)\le k\}\qquad(k\ge0),
\]
so $A_0(D)=0$, and define
\[
 V_0(D)=\{0\},\qquad
 V_k(D)=\{A_{k-1}(D)+1,\ldots,A_k(D)\}\quad(k\ge1).
\]
Write $s_k(D):=|V_k(D)|$.  A level $k$ is \emph{flexible} if
\[
 s_k(D)\ge2\qquad\text{and}\qquad s_{k+1}(D)>s_k(D).
\]
When $D$ is fixed we abbreviate $A_k,V_k,s_k$.
\end{definition}

If $D=D_g$, the threshold law gives $A_k=g^{[k]}(0)$, so the level sizes are
determined by $D$ alone.

\begin{lemma}[depth-preserving level assignments]
\label{lem:level-assignments}
Let $D$ be generator-realizable.  Suppose that for every $k\ge0$ positive
integers $(c_z)_{z\in V_k}$ are chosen so that
\[
 \sum_{z\in V_k}c_z=s_{k+1}.
\]
Define $g(x)=\sum_{z=0}^{x}c_z$.  Then $g$ is a step generator and $D_g=D$.
Conversely, every generator $g$ with $D_g=D$ has positive child counts
$c_0=g(0)$ and $c_z=g(z)-g(z-1)$ for $z>0$, satisfying these identities.
\end{lemma}

\begin{proof}
The positive $c_z$ make $g$ strictly increasing and give $g(x)\ge x+1$.
For $k=0$, $g(A_0)=c_0=s_1=A_1$; for $k\ge1$, telescoping gives
\[
 g(A_k)-g(A_{k-1})=s_{k+1}=A_{k+1}-A_k.
\]
Thus $g(A_k)=A_{k+1}$ for all $k$, hence $g^{[k]}(0)=A_k$ and $D_g=D$.
Conversely, if $D_g=D$, then $A_k=g^{[k]}(0)$; the $k=0$ identity is
$c_0=g(0)=A_1=s_1$, while for $k\ge1$ telescoping gives
\[
 \sum_{z\in V_k}c_z=g(A_k)-g(A_{k-1})=s_{k+1}.
\]
Strict increase makes all $c_z$ positive.
\end{proof}

\begin{lemma}[continuum supply at one exact depth]
\label{lem:continuum-supply}
If $D$ has infinitely many flexible levels, there are $2^{\aleph_0}$ distinct
generators $g$ with $D_g=D$.
\end{lemma}

\begin{proof}
At a flexible level $k$, put $P=A_{k-1}+1$, $Q=P+1$, and
$r=s_{k+1}-s_k\ge1$.  Two admissible assignments are
\[
 c_P=r+1,\quad c_z=1\ (z\ne P),
\]
and
\[
 c_P=1,\quad c_Q=r+1,\quad c_z=1\ (z\notin\{P,Q\}).
\]
At every other level use any fixed admissible assignment.  Independent binary
choices on infinitely many flexible levels give $2^{\Nat}$ distinct cumulative
generators, all of depth $D$ by Lemma~\ref{lem:level-assignments}.
\end{proof}

\begin{definition}[algebraic multiplicity of a depth profile]
For an initial basis $B$ and a generator-realizable depth map $D$, define
\[
 \mathfrak M_B(D)
 :=\left|\left\{\Alg{\rho_g}{B}:D_g=D\right\}\right|.
\]
For the Grzegorczyk bases write
$\mathfrak M_m(D):=\mathfrak M_{B_m}(D)$.
\end{definition}

\begin{theorem}[maximal algebraic multiplicity of an exact-depth fibre]
\label{thm:flexible-continuum}
Let $B$ be any countable initial basis containing zero and the projections, and
let $D$ be a generator-realizable depth profile with infinitely many flexible
levels.  Then
\[
 \boxed{\mathfrak M_B(D)=2^{\aleph_0}.}
\]
In particular, for every $m\in\Nat$,
\[
 \boxed{\mathfrak M_m(D)=2^{\aleph_0}}.
\]
\end{theorem}

\begin{proof}
Lemma~\ref{lem:continuum-supply} gives continuum many distinct generators.
Partition them by equality of the algebras $\Alg{\rho_g}{B}$.  By
Lemma~\ref{lem:descent-internal}, every associated descent in one part is a
unary function in that one algebra.  The algebra is countable by
Lemma~\ref{lem:countable}, and Lemma~\ref{lem:generator-injective} shows that
distinct generators have distinct descents.  Hence each part is countable.
A continuum-sized set cannot be partitioned into fewer than continuum many
countable parts.  The reverse inequality is immediate.
\end{proof}

\begin{corollary}[binary exact-depth multiplicity]
\label{cor:binary-continuum}
Let $D_{\rm bin}(y)=\lceil\log_2(y+1)\rceil$.  Then for every $m$,
\[
 \mathfrak M_m(D_{\rm bin})=2^{\aleph_0}.
\]
\end{corollary}

\begin{proof}
For binary depth, $s_0=1$ and $s_k=2^{k-1}$ for $k\ge1$, so every $k\ge2$
is flexible.  Apply Theorem~\ref{thm:flexible-continuum}.
\end{proof}

\section{Effective algebraic multiplicity}
\label{sec:effective-multiplicity}

Call an initial basis $B$ \emph{effective} if it is given by an effective
enumeration $(b_e)_{e\in\Nat}$ together with a computable arity map
$e\mapsto\operatorname{arity}(b_e)$ and a total computable uniform evaluator
for the values $b_e(\bar x)$ on tuples of the declared arity.  Assume also that
$B$ contains zero and the projections.  For a computable generator-realizable
depth profile $D$, define
\[
 \mathfrak M_B^{\rm eff}(D)
 :=\left|\left\{\Alg{\rho_g}{B}:
 g\text{ computable and }D_g=D\right\}\right|.
\]
For $B=B_m$ write $\mathfrak M_m^{\rm eff}(D)$.

If $D$ is computable, then $A_k$ and $s_k$ are computable.  Indeed, a
generator-realizable depth profile is nondecreasing and unbounded, and
$A_k+1$ is the least $y$ with $D(y)>k$; a direct search therefore terminates.
Then $s_0=1$ and $s_k=A_k-A_{k-1}$ for $k\ge1$.  Thus flexibility is a
decidable property of the level index, and the flexible levels have a
computable increasing enumeration whenever infinitely many occur.

\begin{lemma}[effective avoidance inside one flexible exact-depth fibre]
\label{lem:effective-avoidance}
Let $D$ be a computable generator-realizable depth profile with infinitely many
flexible levels, and let $F_0,F_1,\ldots$ be a uniformly computable sequence of
total unary functions.  There is a computable generator $g$ such that
\[
 D_g=D
\]
and
\[
 \rho_g\ne F_i
 \qquad\text{for every }i\in\Nat.
\]
\end{lemma}

\begin{proof}
Let $k_0<k_1<\cdots$ be the computable enumeration of the flexible levels.
We specify the positive child counts $c_z$ level by level.  At the root set
$c_0=s_1$.  On every nonroot level $k$ not among the $k_i$, use the canonical
left-heavy composition
\[
 c_P=s_{k+1}-s_k+1,
 \qquad
 c_z=1\quad(z\in V_k\setminus\{P\}),
\]
where $P=A_{k-1}+1$ is the first parent of $V_k$.

At the designated level $k_i$, put
\[
 P_i=A_{k_i-1}+1,
 \qquad
 Q_i=P_i+1,
 \qquad
 y_i=A_{k_i}+2,
\]
and let
\[
 r_i=s_{k_i+1}-s_{k_i}\ge1.
\]
Flexibility gives $s_{k_i}\ge2$, so $P_i,Q_i$ are distinct parents and the
second child $y_i$ exists.  If $F_i(y_i)\ne P_i$, choose
\[
 c_{P_i}=r_i+1,
 \qquad
 c_z=1\quad(z\in V_{k_i}\setminus\{P_i\}).
\]
Then the first parent has at least two children and
$\rho_g(y_i)=P_i\ne F_i(y_i)$.  If instead $F_i(y_i)=P_i$, choose
\[
 c_{P_i}=1,
 \qquad
 c_{Q_i}=r_i+1,
 \qquad
 c_z=1\quad(z\in V_{k_i}\setminus\{P_i,Q_i\}).
\]
Now the first child belongs to $P_i$ and the second to $Q_i$, so
\[
 \rho_g(y_i)=Q_i\ne P_i=F_i(y_i).
\]
In either case the child counts are positive and sum to $s_{k_i+1}$.

Define
\[
 g(x)=\sum_{z=0}^{x}c_z.
\]
Lemma~\ref{lem:level-assignments} gives $D_g=D$.  The construction is
computable: from a parent index $z$ one computes its depth level, determines
whether that level is some $k_i$, evaluates at most one value $F_i(y_i)$, and
thereby computes $c_z$; finite summation then computes $g(x)$.  By construction
$\rho_g(y_i)\ne F_i(y_i)$ for every $i$.
\end{proof}

For effective diagonalization, a \emph{raw step term} is a well-typed finite
term over $B$, composition, and step recursion along $\rho$, with the semantic
bound-admissibility test omitted.  Because the bound is not used in the
recurrence, raw terms remain total and form an effectively enumerable superset
of the genuine bounded step algebra.

\begin{lemma}[effective raw-term enumeration]
\label{lem:raw-enumeration}
Fix an effective basis $B$ and a computable descent $\rho$.  The raw unary
step terms admit an effective enumeration
\[
 T_0,T_1,T_2,\ldots
\]
together with a total computable universal evaluator $U(e,x)=T_e(x)$.  Every
unary function in $\Alg{\rho}{B}$ occurs among these denotations.
\end{lemma}

\begin{proof}
The finite well-typed syntax trees over the effective basis $B$, composition,
and one step-recursion constructor are effectively enumerable.  Their
semantics are computed by structural recursion.  At a step-recursion node,
evaluation terminates because each recursive call replaces $y>0$ by
$\rho(y)<y$.  Thus evaluation of the $e$th raw term on input $x$ defines a
total computable universal function $U(e,x)$.  Every admissible bounded
step-recursion derivation is, after forgetting the truth of its bound
certificate, one of these raw terms.
\end{proof}

\begin{theorem}[exact effective multiplicity under persistent flexibility]
\label{thm:effective-multiplicity}
Let $B$ be an effective initial basis, and let $D$ be a computable
generator-realizable depth profile with infinitely many flexible levels.  Then
\[
 \boxed{\mathfrak M_B^{\rm eff}(D)=\aleph_0.}
\]
In particular, for every fixed $m\in\Nat$,
\[
 \boxed{\mathfrak M_m^{\rm eff}(D)=\aleph_0.}
\]
\end{theorem}

\begin{proof}
The upper bound is immediate because there are only countably many computable
generators.

For the lower bound, first choose a computable generator $g_0$ realizing $D$,
for instance the canonical left-heavy choice from
Lemma~\ref{lem:level-assignments}.  Suppose computable generators
$g_0,\ldots,g_{r-1}$ of depth $D$ have already been constructed.  Their
generalized inverses are computable by monotone search.  By
Lemma~\ref{lem:raw-enumeration}, for each $j<r$ the unary functions in
$\Alg{\rho_{g_j}}{B}$ occur in an effectively evaluable sequence of total
computable functions.  Interleave the finitely many sequences to obtain one
uniformly computable list
\[
 F_0,F_1,\ldots
\]
containing every unary function in
\[
 \bigcup_{j<r}\Alg{\rho_{g_j}}{B}.
\]
Apply Lemma~\ref{lem:effective-avoidance} to obtain a computable generator
$g_r$ with $D_{g_r}=D$ and
\[
 \rho_{g_r}\notin
 \bigcup_{j<r}\Alg{\rho_{g_j}}{B}.
\]
By Lemma~\ref{lem:descent-internal}, however,
\[
 \rho_{g_r}\in\Alg{\rho_{g_r}}{B}.
\]
Thus the new algebra differs from every earlier one.  Induction produces
infinitely many pairwise distinct computably generated algebras, proving the
lower bound $\aleph_0$.
\end{proof}

\begin{corollary}[one-way noncontainment in one exact-depth fiber]
\label{cor:one-way-noncontainment}
Under the assumptions of Theorem~\ref{thm:effective-multiplicity}, there are
computable generators $g_0,g_1,\ldots$ with $D_{g_r}=D$ such that
\[
 \boxed{
 \Alg{\rho_{g_r}}{B}\nsubseteq
 \Alg{\rho_{g_j}}{B}
 \qquad(j<r).}
\]
In particular this holds for each fixed Grzegorczyk basis $B_m$.
\end{corollary}

\begin{proof}
This is immediate from the construction in
Theorem~\ref{thm:effective-multiplicity}: $\rho_{g_r}$ belongs to the later
algebra and to none of the earlier ones.
\end{proof}

\begin{corollary}[binary effective multiplicity]
\label{cor:binary-effective}
For binary depth and every fixed $m$,
\[
 \boxed{\mathfrak M_m^{\rm eff}(D_{\rm bin})=\aleph_0.}
\]
\end{corollary}

\begin{remark}
The effective construction guarantees computability, not a low complexity
bound; the explicit plateau generator supplies a separate low-complexity
witness.
\end{remark}

\section{Discussion and conclusion}

Exact depth fixes the level of every input and, by
Proposition~\ref{prop:depth-orbit}, every orbit endpoint, but it does not record
which parent in the preceding level is chosen.  The explicit plateau example
shows that this missing geometry survives bounded step-recursion closure, while
Theorem~\ref{thm:flexible-continuum} shows that the ambiguity can be maximal:
one exact-depth fibre supports continuum many distinct algebras over every
countable basis containing zero and the projections.  The effective theorem
shows that this is not merely a cardinality effect.

Thus no invariant factoring only through stage counts or orbit endpoints can
classify step-recursion expressiveness.  A natural next problem, for fixed
basis $B_m$ and computable $D$, is to classify equality and inclusion inside
\[
 \{\Alg{\rho_g}{B_m}:g\text{ computable and }D_g=D\}.
\]
The target is a basis-dependent invariant finer than $D_g$ but coarser than the
full predecessor map, for example one built from child-count distributions or
other finite summaries of within-level parent geometry.

The main point is therefore not that depth estimates require finer asymptotics,
but that exact depth itself quotients away algebraically visible information.
At basis zero the source--offset theorem identifies one such obstruction:
fixed-stride halving cannot reproduce a macroscopic concentration of children
inside a level.  Exact depth should consequently be treated as one coordinate
of recursion power, not as a complete semantic invariant.

\section*{Acknowledgements}
The author acknowledges the use of large language models for copyediting,
grammatical correction, and language polishing.  The author reviewed the
resulting text and takes full responsibility for the final manuscript.

\end{document}